\documentclass[12pt]{article}

\usepackage{amstext,amsmath,amssymb,amsfonts,bbm}
\usepackage[latin1]{inputenc}
\usepackage{epsfig}
\usepackage{dsfont}
\usepackage{hyperref}
\usepackage{amsthm}
\usepackage{subfigure}
\usepackage{color}
\usepackage{multirow}
\usepackage{psfrag}
\usepackage{graphicx}
\usepackage{extarrows}

\theoremstyle{plain}

\newtheorem{lemma}{Lemma}

\newtheorem{theorem}{Theorem}
\newtheorem{corollary}{ Corollary }

\newtheorem{proposition}{Proposition}
\theoremstyle{definition}

\begin{document}
 
\title{Wigner law for matrices with dependent entries - 
a perturbative 
approach}
\author{T. Krajewski, A. Tanasa, D. L. Vu}

\maketitle
\abstract{We show that Wigner semi-circle law holds for Hermitian matrices with dependent entries,  provided the deviation of the cumulants from the normalised Gaussian case obeys a simple power law bound in the size of the matrix.
To establish this result, we use replicas interpreted as a zero-dimensional quantum field theoretical model whose effective potential obey a renormalisation group equation.
}

\medskip


\section{Introduction}

In the last decade or so, several extensions of Wigner law for matrices with dependent entries have been considered, see for instance \cite{dependent1}. In this paper, the authors impose bounds on the number of entries of the matrices that are correlated. In \cite{dependent2}, a generalisation of this result was proven. Our approach is complementary: we do not impose such restrictions but we assume that the size of the correlations go to zero as $N$, the size of the matrix, becomes large. Another result concerning Wigner law for matrices with dependent entries have been obtained in  \cite{dependent3}, where the matrices considered are real-valued, symmetric and have  stochastically independent diagonals.

To establish our result we use the replica method. For a standard reference on the use of replica techniques in the context of random matrix theory, we refer the interested reader to \cite{replicas-book}, see also \cite{replica} or \cite{replicas2}.  In this paper we interpret replicas as fields of a zero-dimensional quantum field theoretical model and use an analogue of the renormalisation group equation.  We give conditions on the joint cumulants of the entries of the matrix, see Theorem \ref{Wignerdependent:thm}, under which the moments of the eigenvalue distribution  converges towards the Wigner semi-circle law. 

For the sake of completeness, let us also mention that Wigner law has been studied by mathematical physicists using the so-called supersymmetric technique (see, for example, \cite{ami} and \cite{ami2}), technique which uses calculus over commuting and anti-commuting variables.

In the first part of the paper, we present our result and illustrate it for Wigner matrices and for invariant matrices. The second part is devoted to the proof.

\section{Statement of the main result; a few illustrations}

\subsection{Semi-circle law from a condition on cumulants}

Let us consider a probability law on Hermitian $N\times N$ matrices given by the joint probability density  $\rho_{N}$ for the real diagonal elements and the complex upper diagonal ones, the lower diagonal ones being recovered by complex conjugation.  We assume that the joint cumulants exist for all $N$ and collect them in their generating function defined as
\begin{equation}
\log\langle\exp \text{Tr}(MJ)\rangle=\log\int dM \rho_{N}(M)\exp \text{Tr}(MJ)
\end{equation}
where $dM$ is the product of the Lebesgue measures on the entries of $M$. The source $J$ is another Hermitian $N\times N$ matrix and the cumulants are obtained by derivation with respect to $J$ at the origin,
\begin{equation}
\frac{\partial}{\partial J_{j_{1}i_{1}}}\dots\frac{\partial}{\partial J_{j_{k}i_{k}}}\log\langle\exp \text{Tr}(MJ)\rangle\Big|_{J=0}=\langle M_{i_{1}j_{1}}\cdots M_{i_{k}j_{k}}\rangle_{\text{c}},
\end{equation}
where we used the subscript c because in field theory they correspond to connected correlation functions. Alternatively, the moments are denoted by $\langle M_{i_{1}j_{1}}\cdots M_{i_{k}j_{k}}\rangle$.

Let us associate to each cumulant an oriented graph $G$ constructed as follows. The vertices of $G$ are given by the distinct indices appearing in the cumulant. Let us emphasize that it is fundamental that the indices attached to the vertices are all different, this can be achieved by inserting $1=\delta_{ij}+1-\delta_{ij}$ for every pair of indices. We draw an arrow oriented from the vertex associated to $i$ to the vertex associated to $j$ if the cumulant involves a matrix element $M_{ij}$, see figure \ref{cumulant:fig} for a few examples.
\begin{figure}

\begin{equation*}
\begin{array}{cc}
\includegraphics[width=4cm]{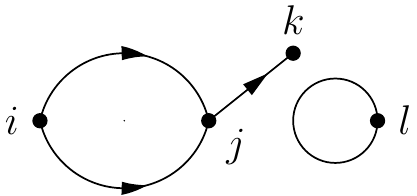}\qquad&\qquad\includegraphics[width=2cm]{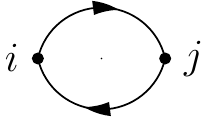}\\
\langle (M_{ij})^{2}M_{jk}M_{ll}\rangle_{\text{c}}\qquad&\qquad\langle M_{ij}M_{ji}\rangle_{\text{c}}
\end{array}
\end{equation*}
\caption{Oriented graphs associated to some cumulants}
\label{cumulant:fig}
\end{figure}
Then, we consider the cumulants as functions $C_{G}$ on the vertex indices, labelled by the graphs $G$,  $i_{1},..., i_{v(G)}\in\{1,\dots,N\}^{v(G)}\mapsto C_{G}(i_{1},..., i_{v(G)})$, with $v(G)$ the number of vertices of $G$.
The generating function of the cumulants can be written as
\begin{equation}
\log\langle\exp \text{Tr}(MJ)\rangle=
\sum_{G\text{ oriented graph}}\frac{1}{|\text{Aut}(G)|}\sum_{1\leq i_{1},\dots, i_{v(G)}\leq N\atop\text{all different}}C_{G}(i_{1},\dots, i_{v(G)})\prod_{e\text{ edge}}J_{i_{s(e)}
i_{t(e)}}\label{cumulantexpansion:eq}
\end{equation}
where $s(e)$ and $t(e)$ are the source and the target vertices of $e$. For a fixed graph, we consider $C_{G}$ as a function of the indices attached to the vertices. For example, in the Gaussian case $\rho(M)\propto\exp-\big(\frac{1}{2\alpha^{2}}\text{Tr}\,M^{2}\big)$, the only non vanishing cumulant is $\langle M_{ij} M_{kl}\rangle_{\text{c}}=\alpha^{2}\delta_{il} \delta_{jk}$ and corresponds either to an oriented cycle with two vertices (when $i\neq j$) or to a graph with a single vertex and two edges (when $i=j$).

In this letter, we are interested in the normalised density of  eigenvalues of a random $N\times N$ Hermitian matrix 
\begin{equation}
   \rho_{N}(\lambda)=
   \frac{1}{N}\sum_{1\leq i\leq N}
   \big\langle\delta(\lambda-\lambda_{i})\big\rangle,
\end{equation}
where $\lambda_{i}$ are the eigenvalues of the random matrix $\frac{M}{\sqrt{N}}$ and $\delta$ the Dirac distribution. The average is computed using a probability density $\rho_{N}(M)$. Note that we have used the same letter to denote the density of eigenvalues $\rho_{N}(\lambda)$ and the probability law on the space of matrices $\rho_{N}(M)$ in order to simplify the notations.

To state our main result, we impose different conditions on the cumulants, depending on whether $G$ is Eulerian or not. Recall that an oriented graph $G$ is Eulerian if every vertex of $G$ has an equal number of incoming and outgoing edges 

\begin{theorem}
\label{Wignerdependent:thm}
Let $\rho_{N}$ be a probability law on the space of Hermitian $N\times N$ matrices $M$ such that its cumulants can be decomposed as 
$C_{G}=C_{G}^{'}+C_{G}^{''}$, with $C_{G}^{'}$ a Gaussian cumulant and $C_{G}^{''}$ a perturbation such that, uniformly in the vertex indices $i_{1},..., i_{v(G)}$,
\begin{itemize}
   \item ${\displaystyle \lim_{N\rightarrow\infty}\, N^{v(G)-c(G)-e(G)/2}C_{G}(i_{1},..., i_{v(G)})=0}$ if $G$  is Eulerian,
    \item ${\displaystyle N^{v(G)-c(G)-e(G)/2}C_{G}(i_{1},..., i_{v(G)})}$ is bounded if $G$ is not Eulerian,
\end{itemize}
where $v(G)$, $e(G)$, $c(G)$ are the number of vertices, edges and connected components of $G$. Then, the moments of the eigenvalue distribution of the matrix $\frac{M}{\sqrt{N}}$ converge towards the moments of the semi-circle law, with $\alpha$ given by the Gaussian cumulant  $\langle M_{ij} M_{kl}\rangle_{\text{c}}=\alpha^{2}\delta_{il} \delta_{jk}$,
\begin{equation}
   \lim_{N\rightarrow\infty} \int_{\mathbb R} d\lambda \,\lambda^{k}\rho_{N}(\lambda)
   =\begin{cases}
   \frac{1}{2\pi \alpha^{2}}\int_{-2\alpha}^{2\alpha} d\lambda\,\lambda^{k}\sqrt{4\alpha^{2}-\lambda^{2}}&\text{if $k$ is even},\\
0&\text{if $k$ is odd}.\label{semicircle:eq}
\end{cases}
\end{equation}

\end{theorem}

The conditions are uniform in the sense that they must not depend on the vertex indices but can depend on the graph. These conditions can be understood heuristically using the moment method proof of Wigner's law (see for instance \cite{Guionnet})
\begin{equation}
\int_{\mathbb R} d\lambda \,\lambda^{k}\rho_{N}(\lambda)=
\frac{1}{N^{1+k/2}}\langle \text{Tr\,}M^{k}\rangle=\sum_{1\leq i_{1},\dots,i_{k}\leq N}\langle M_{i_{1}i_{2}}M_{i_{2}i_{3}}\cdots M_{i_{k}i_{1}}\rangle.
\end{equation}
Next, we express the moments in terms of the cumulants and label the latter by the oriented graphs $G$. The trace imposes the existence of an Eulerian cycle in these graphs, which always exists for a connected graph such that all its vertices have an equal number of incoming and outgoing edges. If not, some identifications of vertices are necessary, i.e. the vertex indices $i_{1},\dots,i_{k}$ cannot all be distinct. Then, the conditions on the cumulants imply that only the contribution of the Gaussian one survives in the limit $N\rightarrow\infty$. Also note that in the case of real symmetric matrices, the condition for Eulerian graphs apply to all graphs.

\subsection{Illustration 1: Semi-circle law for Wigner matrices}

Recall that for Wigner matrices the diagonal elements are independent and identically distributed (iid =)and the real and imaginary parts of the upper diagonal elements are also iid, independent from the diagonal ones but possibly with a different law, such that the expectation value of the off diagonal elements vanish. We further assume that all moments (thus also cumulants) remain finite as $N$ becomes large. Under these assumptions, Wigner's seminal result follows from Theorem \ref{Wignerdependent:thm}.
\begin{corollary}
For Wigner matrices with finite moments, the eigenvalue distribution converges in moments towards the semi-circle law \eqref{semicircle:eq} when $N\rightarrow\infty$, with $\alpha$ given by the Gaussian cumulant.
\end{corollary}

\begin{proof}
Since the matrix elements are independent, all cumulants with $v\geq 3$  vanish. Therefore, we have to check the bounds of Theorem \ref{Wignerdependent:thm} for $v=1,2$ only. For $v=1$ the condition is obviously satisfied because of the factor $N^{e/2}$ and the fact that moments and therefore also cumulants are bounded. Furthermore, if $v=2$ and $e\geq 3$, the condition is also satisfied, for the same reason. The case $v=2$ and $e=1$ corresponds to the expectation values of the off-diagonal terms and vanish identically. 
The case $v=2$ and $e=2$ remains to be studied:
\begin{itemize}
    \item With $c=2$, $C_{\includegraphics[width=1cm]{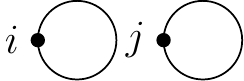}}(i,j)=\langle M_{ii}M_{jj}\rangle_{c}=
       \langle M_{ii}M_{jj}\rangle-
       \langle M_{ii}\rangle\langle M_{jj}\rangle=0$, because diagonal matrix elements are independent.
       \item   With $c=1$, non Eulerian, $C_{\includegraphics[width=1cm]{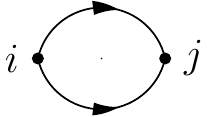}}(i,j)=\langle M_{ij}M_{ij}\rangle_{c}=
       \langle M_{ij}M_{ij}\rangle-
       \langle M_{ij}\rangle\langle M_{ij}\rangle=
       \langle (\Re M_{ij})^{2}-(\Im M_{ij})^{2}\rangle+2\text{i}
       \langle \Re M_{ij}\Im M_{ij}\rangle=0$, since the real and imaginary parts of the off diagonal elements are independent and identically distributed with mean 0.
       \item  With $c=1$, Eulerian, $C_{\includegraphics[width=1cm]{cumulant2}}=\alpha^{2}$ is the Gaussian cumulant leading to the semi-circle law. Note that only the part $i\neq j$ contributes at large $N$.
\end{itemize}\end{proof}

\subsection{Illustration 2: Unitary invariant potential}

Let us consider a single trace unitary invariant potential. In this case, the probability law is  
\begin{align}
    \rho(M)=\frac{\exp-\text{Tr}\,V(M)}{Z}\qquad\text{with}\qquad Z=\int dM \,\exp-\text{Tr}\,V(M)
\end{align}
and $V(M)$ a polynomial potential
\begin{align}
    V(M)=\frac{1}{2}M^{2}+\sum_{p\geq 3}\frac{g_{p}N^{1-p/2}}{p}M^{p}.
\end{align}
Recall that here we study the eigenvalue distribution of the matrix $M/\sqrt{N}$, the usual formulation employed in the physics literature being recovered after the rescaling $M\rightarrow \sqrt{N}\, M$.

Cumulants of order $k$ can be computed as sums of connected ribbon graphs $\Gamma$ (not to be confused with the graphs $G$ appearing in the cumulants) with $k$ univalent vertices corresponding to the insertions of the source $J$. Its dependence on $N$ reads
$N^{f(\Gamma)+\sum_{p}(1-pv_{p}(\Gamma)/2)}$
where $f(\Gamma)$ is the number of closed faces of the graph (closed cycles in the double line representation) and $v_{p}(\Gamma)$ the number of vertices of degree $p$.

Denoting by $e(\Gamma)$ the number of internal edges of the graph $\Gamma$ (edges not attached to the sources), we have $2e(\Gamma)+k=\sum_{p}pv_{p}(\Gamma)$. Moreover, the Euler characteristics of the surface with boundary in which the graph is embedded reads $2-2g(\Gamma)-b(\Gamma)=f(\Gamma)-e(\Gamma)+v(\Gamma)$,
with $v(\Gamma)=\sum_{p}v_{p}(\Gamma)$ the total number of vertices (not including the sources), $g(\Gamma)$ the genus of the surface and $b(\Gamma)$ its number of boundary components (open faces including insertions of the source). Combining these identities together, the power of $N$ in a graph contributing to a cumulant reads
\begin{equation}
    N^{2-2g(\Gamma)-b(\Gamma)-k/2}\label{unitaryscaling:eq}.
\end{equation}
The leading order contribution is obtained for planar graphs ($g=0$) with all sources in the same open face ($b=1$). 

These cumulants correspond to graphs $G$  that are oriented cycles with $e(G)=v(G)=k$ and $c(G)=1$.  According to \eqref{unitaryscaling:eq}, they scale as $N^{v(G)-c(G)-e(G)/2}$ for large $N$, thus violating the first condition on the cumulants in Theorem  \ref{Wignerdependent:thm}. Since it is known that such random matrices do not obey the semi-circle law if $V$ is not Gaussian, we conclude that the conditions in Theorem \ref{Wignerdependent:thm} are optimal, at least when formulated using vertices, edges and connected components of $G$.

\section{Replica proof of the main result}
\subsection{An expression of the Green function using replicas}

Using $\rho(\lambda)=-\frac{1}{\pi}\Im G(\lambda+\text{i}0^{+})$, the density of eigenvalues is determined by the Green function
\begin{equation}
    G(z)=\frac{1}{N}\Big\langle\text{Tr}\Big(z-\frac{M}{\sqrt{N}}\Big)^{-1}\Big\rangle=
    \frac{1}{N}\int dM\,\rho_{N}(M)\,\text{Tr}\Big(z-\frac{M}{\sqrt{N}}\Big)^{-1}
\end{equation}
To compute the resolvent, first notice that $\text{Tr}\big(z-\frac{M}{\sqrt{N}}\big)^{-1}=\frac{\partial}{\partial z}\log\det\big(z-\frac{M}{\sqrt{N}}\big)$. Then, the resolvent is computed using bosonic replicas
 \begin{equation}
   G(z)=-\frac{1}{N} \frac{\partial}{\partial z}
   \bigg(
   \int dX^{\dagger}dX\,
  \exp-z\text{Tr}(X^{\dagger}X)\,
   \Big\langle\exp\text{Tr}\Big(X^{\dagger}\frac{M}{\sqrt{N}}X\Big)\Big\rangle
   \bigg)_{\text{order 1 in $n$}}\label{replica:eq}
\end{equation}
where $X$ is a $N\times n$ complex matrix and $dX^{\dagger}dX$ is the product of Lesbegue measure over real and imaginary parts of $X$ divided by a factor of $\pi^{nN}$. Although the use of the replica technique is common in physics, let us give some explanations. We first express $\det^{-n}\big(z-\frac{M}{\sqrt{N}}\big)$ using a Gaussian integral over $n$ complex vectors, each with $N$ components that we collect in the matrix $X$. Then, its logarithm is computed using the identity $A^{n}=1+n\log A+O(n^{2})$ for $n\rightarrow 0$. In our context, we evaluate the integral \eqref{replica:eq} as a perturbation of a Gaussian integral using Feynman diagrams, thus leading to a power series in $1/z$. Because of the $\text{U}(n)$ invariance $X\rightarrow XU$ of the integral in \eqref{replica:eq}, each Feynman amplitude is a polynomial in $n$ and we retain only the term of order $n$. In perturbation theory, this is nothing but a convenient substitute for the power series of $\log\det\big(z-\frac{M}{\sqrt{N}}\big)$. Here, we stick to the perturbative approach but it is worthwhile to mention that, beyond perturbation theory, one encounters the phenomenon of replica symmetry breaking, ruining the simple polynomial dependence on $n$, see for instance \cite{Mezard}.

\subsection{Diagrammatic approach for the Gaussian case}

Before proceeding to the general case and establish Theorem \ref{Wignerdependent:thm}, 
let us consider the Gaussian case
$\rho(M)\propto\exp-\frac{\text{Tr}(M^{2})}{2\alpha^{2}}$. Performing the integral over $M$ leads to a quartic interaction for the replicas
\begin{equation}
    \Big\langle\exp\text{Tr}\Big(X^{\dagger}\frac{M}{\sqrt{N}}X\Big)\Big\rangle=
    \exp\frac{\alpha^{2}}{2N}\text{Tr}(X^{\dagger}XX^{\dagger}X).\label{Gaussian:eq}
\end{equation}
Then, we expand the integral in \eqref{replica:eq} using Feynman diagrams. The latter are ribbon graphs with double lines made of a solid line for the matrix indexes $i,j,k,...\in\left\{1,...,N\right\}$ and a dotted line for the replica indexes $a,b,c,...\in\left\{1,...,n\right\}$. The graphs that correspond to the terms of order $n$ are those  with a single dotted face. Moreover, when we take the large $N$ limit, we only retain the planar graphs. Each solid face yields a factor of $N$ which is cancelled by the factor of $1/N$ in the vertices and in front of the integral \eqref{replica:eq}. Note that it is crucial to select the term of order 1 in $n$ before taking the large $N$ limit. In this example, it is also easy to see that at each order in $1/z$ we have only a finite number of graphs, thus yielding a polynomial in $n$. They consist in several solid faces such that each two faces share at most one vertex. Moreover, these faces are enclosed in a single dotted face, in particular all these graphs are connected.

\begin{figure}
\begin{equation*}
\begin{array}{cc}
\parbox{2.8cm}{\includegraphics[width=2.8cm]{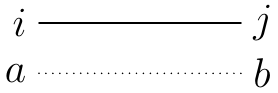}}
\quad\rightarrow\quad\displaystyle{\frac{1}{z}}\delta_{ij}\delta_{ab}\qquad&
\qquad\parbox{3cm}{\includegraphics[width=3cm]{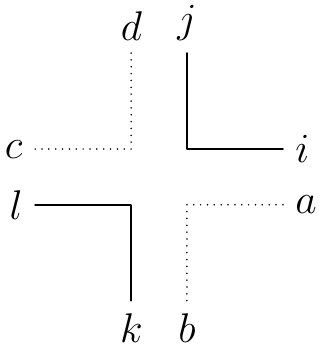}}
\quad\rightarrow\quad\displaystyle{\frac{a^{2}}{2}}\delta_{ij}\delta_{kl}\delta_{ab}\delta_{cd}\qquad
\end{array}
\end{equation*}
\caption{Feynman rules in the Gaussian case}
\label{FeynmanGaussian:fig}
\end{figure}

Furthermore, taking the derivative with respect to $z$ inserts a cilium on one of the edges and removes all symmetry factors. The sum over graphs with a cilium is nothing but $G(z)$ and obeys the equation, in the large $N$ limit,
\begin{equation}
    G(z)=\sum_{k=0}^{\infty} \frac{\alpha^{k}G^{k}(z)}{z^{k+1}}=\frac{1}{z-\alpha  G(z)}.
\end{equation}
Indeed, if we remove the solid face containing the cilium, we get $k$ copies of $G(z)$, where $k$ is the number of vertices of that face, see figure \ref{Gaussian_decomposition:fig}. The solution that behaves as $1/z$ for large $z$ is
\begin{equation}
    G(z)=\frac{z}{2\alpha^{2}}\bigg(1-\sqrt{1-\frac{4\alpha^{2}}{z^{2}}}\bigg).
\end{equation}
Finally, from the cut of the square root on the negative real axis, we obtain Wigner semi-circle law \eqref{semicircle:eq} in the large $N$ limit.

\begin{figure}
\begin{equation*}
\parbox{5cm}{\includegraphics[width=5cm]{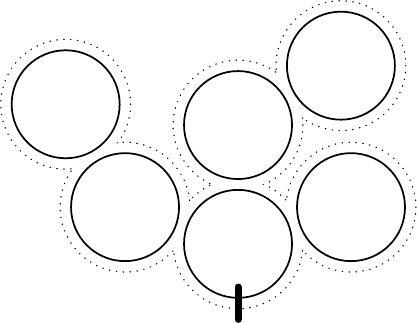}}
\qquad\rightarrow\qquad
\parbox{1.5cm}{\includegraphics[width=1.5cm]{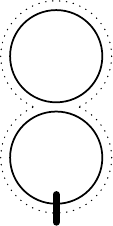}}\quad
\parbox{1.5cm}{\includegraphics[width=1.5cm]{Gaussian_replicas_final1}}\quad
\parbox{1.5cm}{\includegraphics[width=1.5cm]{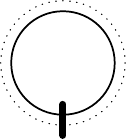}}
\end{equation*}
\caption{A graph to $G(z)$ in the large $N$ limit and its decomposition}
\label{Gaussian_decomposition:fig}
\end{figure}

\subsection{Cumulants and the replica effective action}

When we substitute $J=\frac{XX^{\dagger}}{\sqrt{N}}$, the generating function of the cumulants plays the role of an interacting potential for the replicas $V_{0}(X,X^{\dagger})=\log\big\langle\exp \text{Tr}\big(M\frac{XX^{\dagger}}{\sqrt{N}}\big)\big\rangle$, so that  \eqref{replica:eq}
writes
\begin{equation}
   G(z)=-\frac{1}{N} \frac{\partial}{\partial z}
   \bigg(
   \int dX^{\dagger}dX\,
  \exp\Big\{-z\text{Tr}(X^{\dagger}X)\,
   +V_{0}(X,X^{\dagger})\Big\}
   \bigg)_{\text{order 1 in $n$}}\label{replicapotential:eq}
\end{equation}
Although the expansion of $V_{0}(X,X^{\dagger})$ in powers of $X,X^{\dagger}$ (see \eqref{cumulantexpansion:eq}) involves a power series, at any order in $1/z$, only a finite number of cumulants $C_{G}$ appear. This expansion is based on Feynman diagrams, treating 
$V_{0}(X,X^{\dagger})$ as a perturbation.

Motivated by the quantum field theory analogy, we introduce the effective potential
\begin{equation}
    V(t; X, X^{\dagger})=\log\int d Y^{\dagger}d Y\,
  \exp\left\{-\frac{\text{Tr}( Y^{\dagger} Y)}{t}+V_{0}( X+ Y, X^{\dagger}+ Y^{\dagger})\right\}-Nn\log\,t
  \label{effective:eq}
\end{equation}
where we have set $t=1/z$ for later convenience. The normalisation of the measure only contains a factor of $\pi^{Nn}$ and we added an extra contribution of $-Nn\log\,t$ in such a way that the Gaussian integral is fully normalized. At $t=0$ (equivalently for $z\rightarrow\infty$), we have $V(t=0; X, X^{\dagger})=V_{0}( X, X^{\dagger})$, since there is no integration over $ Y$ and $ Y^{\dagger}$.

The effective potential obeys the Gaussian convolution identity, see \cite{Zinn},
\begin{equation}
    V(t+s; X, X^{\dagger})=\log\int d Y^{\dagger}d Y\,
  \exp\left\{-\frac{\text{Tr}( Y^{\dagger} Y)}{s}+V( t;X+ Y, X^{\dagger}+ Y^{\dagger})\right\}-Nn\log\, s\end{equation}
Expanding to first order in $s$ yields the differential equation, see also \cite{Zinn},
\begin{equation}
    \frac{\partial V}{\partial t}=
    \sum_{i,a}
    \Bigg(
    \frac{\partial^{2}V}{\partial X_{i,a}\partial\overline{ X}_{i,a}}+
    \frac{\partial V}{\partial X_{i,a}}
    \frac{\partial V}{\partial\overline{ X}_{i,a}}
    \Bigg),
    \label{Polchinski:eq}
\end{equation}
where $1\leq i\leq N$ is a matrix index and  $1\leq a\leq n$ a replica index. This is nothing but a zero dimensional analogue of the Polchinski  exact renormalisation group equation \cite{Polchinski} that governs the flows of effective actions in quantum field theory. Equivalently, \eqref{Polchinski:eq} can be written in integral form,
\begin{equation}
    V(t;X, X^{\dagger})=V_{0}( X, X^{\dagger})+\int_{0}^{t}ds\,
    \sum_{i,a}
    \Bigg(
    \frac{\partial^{2}V(s; X, X^{\dagger})}{\partial X_{i,a}\partial\overline{ X}_{i,a}}+
    \frac{\partial V(s; X, X^{\dagger})}{\partial X_{i,a}}
    \frac{\partial V(s; X, X^{\dagger})}{\partial\overline{ X}_{i,a}}
    \Bigg).
    \label{integral:eq}
\end{equation}
Solving this integral equation iteratively proves to be helpful to establish results order by order in powers of $t=\frac{1}{z}$.

The Green function $G(z)$ and thus also the eigenvalue density can be expressed in terms of the effective potential. Indeed, let us rewrite \eqref{effective:eq} as
\begin{equation}
\int d Y^{\dagger}d Y\,
  \exp\left\{-z\text{Tr}( Y^{\dagger} Y)+V_{0}( Y, Y^{\dagger})\right\} =z^{-nN} \exp V(t=1/z; X=0, X^{\dagger}=0).
  \end{equation}
Deriving this equation with respect to $z$ and using \eqref{Polchinski:eq} and  \eqref{replica:eq}, we obtain
\begin{equation}
   G(z)=\frac{1}{z}+\frac{1}{Nz^{2}}
   \Bigg[\sum_{i,a}
    \bigg(
    \frac{\partial^{2}V(1/z; 0, 0)}{\partial X_{i,a}\partial\overline{ X}_{i,a}}+
    \frac{\partial V(1/z; 0, 0)}{\partial X_{i,a}}
    \frac{\partial V(1/z; 0, 0)}{\partial\overline{ X}_{i,a}}
    \bigg)
   \Bigg]_{\text{order 1 in $n$}}
\end{equation}
In order to compute the RHS of \eqref{replicapotential:eq}, let us expand it in powers of $X, X^{\dagger}$.

The effective potential can be developed on  oriented graphs $G$ as in \eqref{cumulantexpansion:eq}, which we recover for $t=0$,
\begin{equation}
V(t; X, X^{\dagger})=
\sum_{G\text{ oriented graph}}\frac{1}{|\text{Aut}(G)|N^{e(G)/2}}\sum_{1\leq i_{1},\dots,i_{v(G)}\leq N\atop\text{all different}}C_{G}(t;i_{1},\dots, i_{v(G)})\prod_{e\text{ edge}}(XX^{\dagger})_{i_{s(e)}
i_{t(e)}}\label{effectiveexpansion:eq}.
\end{equation}
The oriented graphs $G$ in the previous equation should not be confused with Feynman diagrams appearing in  a perturbative computation, they are merely labels for the terms in the generating functions of the cumulants and in the effective potential.

In particular, only the first non trivial term in this expansion contributes to the Green function, as seen from \eqref{replicapotential:eq},
\begin{equation}
    G(z)=\frac{1}{z}+\frac{1}{N^{3/2}z^{2}}\sum_{1\leq i\leq N}\big[C_{\includegraphics[width=0.8cm]{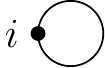}}(1/z;i) \big]_{\text{order 0 in $n$}}.\label{Greenpotential:eq}
 \end{equation}
Notice that only the order 0 in $n$ is necessary since an extra power of $n$ is created by the summation over $a$.

\subsection{Inductive bounds on Green functions}
In order to establish Theorem \ref{Wignerdependent:thm}, let us recall that the cumulants are written as a sum of a Gaussian term and a perturbation. This translates into a similar decomposition for the potential $V(t=0;X, X^{\dagger})=V^{'}(t=0;X, X^{\dagger})+V^{''}(t=0;X, X^{\dagger})$, with $V^{'}(t=0;X, X^{\dagger})=\frac{\alpha^{2}}{2N}\text{Tr}(X^{\dagger}XX^{\dagger}X)$ the quartic interaction derived from the Gaussian cumulant and $V^{''}(t=0;X, X^{\dagger})$ a perturbation.

Then, solving iteratively the integral equation \eqref{integral:eq} as a power series in $t$ yields a similar decomposition for the effective potential $V(t;X, X^{\dagger})=V^{'}(t;X, X^{\dagger})+V^{''}(t;X, X^{\dagger})$, where $V^{'}(t;X, X^{\dagger})$ only involves the quartic interaction  $V^{'}(t=0;X, X^{\dagger})$ whereas $V^{''}(t;X, X^{\dagger})$ contains at least one occurrence of the perturbation, in particular, it vanishes if 
$V^{''}(t=0;X, X^{\dagger})=0$.

Furthermore, let us decompose both terms at order 0 in $n$ using oriented graphs as in \eqref{effectiveexpansion:eq}, so that, as a power series in $t$, 
\begin{equation}
C_{G}(t;i_{1},\dots, i_{v(G)})=\sum_{k=0}^{\infty}t^{k} \big[C^{'(k)}_{G}(t;i_{1},\dots, i_{v(G)})+C^{''(k)}_{G}(t;i_{1},\dots, i_{v(G)})\big].
\label{graphexpansion:eq}
\end{equation}

Using \eqref{integral:eq}, we show by induction on $k$ that the perturbation obeys a bound identical to the assumption of Theorem \ref{Wignerdependent:thm}.

\begin{proposition} \label{bound:pro} The coefficients of the development of the effective potential over graphs satisfy, uniformly in $i_{1},..., i_{v(G)}$,
\begin{itemize}
\item ${\displaystyle N^{v(G)-c(G)-e(G)/2}\big[C^{'(k)}_{G}(i_{1},..., i_{v(G)})}\big]_{\text{order 0 in $n$}}$ is bounded for any  $G$,
   \item ${\displaystyle \lim_{N\rightarrow\infty}\, N^{v(G)-c(G)-e(G)/2}\big[C^{''(k)}_{G}(i_{1},..., i_{v(G)})=0}\big]_{\text{order 0 in $n$}}$ if $G$ is Eulerian,    
   \item ${\displaystyle N^{v(G)-c(G)-e(G)/2}\big[C^{''(k)}_{G}(i_{1},..., i_{v(G)})}\big]_{\text{order 0 in $n$}}$ is bounded if $G$ is not Eulerian.
\end{itemize}
\end{proposition}
\begin{proof}

At order 0, the conditions are satisfied by the Gaussian cumulant \eqref{Gaussian:eq} and the non Gaussian ones since these are just the assumptions of Theorem \ref{Wignerdependent:thm}. Let us assume that the conditions hold up to order $k-1$ and use \eqref{integral:eq} to show that they also hold at order $k$.

The derivative with respect to $X_{i,a}$  (resp. $\overline{X}_{i,a}$) acting the graph expansion  \eqref{graphexpansion:eq} removes an outgoing (resp. incoming) half line attached to a vertex labelled $i$. This operation is performed either on a single graph (first term in \eqref{integral:eq}) or on two independent graphs (second term in \eqref{integral:eq}). The subsequent summation over $i$ and $a$ reattaches the remaining half lines, see picture \ref{graphoperations:fig}. Collecting all contributions to a graph appearing on  the LHS of \eqref{integral:eq}  (order $k$) to those  appearing on its RHS (order $<k$) allows us to express an order $k$ term using order $<k$ terms, all to order 0 in $n$. In the sequel, we denote by $d_{+}(v)$ (resp. $d_{-}(v))$ the number of incoming (resp. outgoing) edges to a vertex $v$.

\begin{figure}
\begin{equation*}
\begin{array}{cccc}
{\displaystyle\frac{\partial^{2}V}{\partial X_{i,a}\partial\overline{ X}_{i,a}}}\,:\quad&\parbox{2.8cm}{\includegraphics[width=2.8cm]{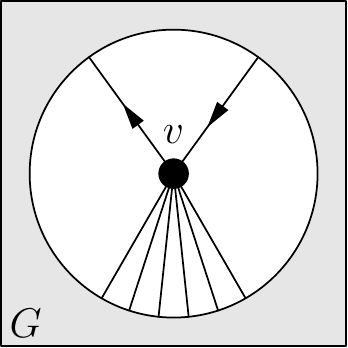}}
&\quad\rightarrow\quad&
\parbox{2.8cm}{\includegraphics[width=2.8cm]{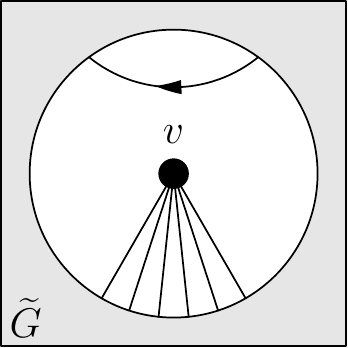}}\\
\\
{\displaystyle\frac{\partial V}{\partial X_{i,a}}
    \frac{\partial V}{\partial\overline{ X}_{i,a}}}\,:\quad
    &\parbox{2.8cm}{\includegraphics[width=2.8cm]{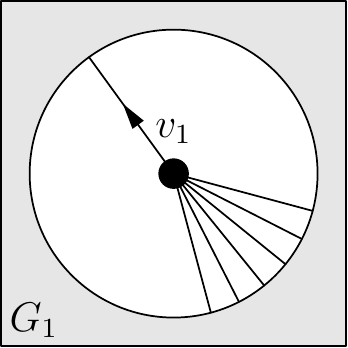}}
\quad\parbox{2.8cm}{\includegraphics[width=2.8cm]{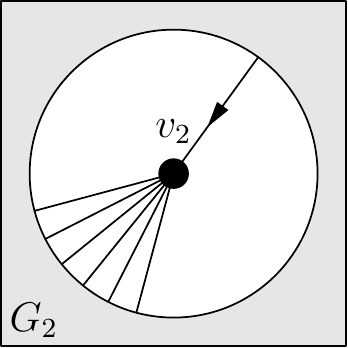}}&\quad\rightarrow\quad&
\parbox{2.8cm}{\includegraphics[width=2.8cm]{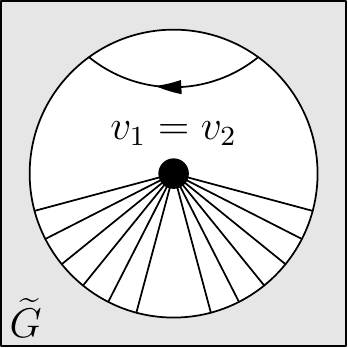}}
\end{array}
\end{equation*}
\caption{Graphical interpretation of the RHS of \eqref{integral:eq}}
\label{graphoperations:fig}
\end{figure}

Let us consider the first operation on a vertex $v$ of a graph $G$ on  the RHS of \eqref{integral:eq} and denote by $\widetilde{G}$ the resulting graph. First note that this operation does not change the nature of the graph (Gaussian or not, Eulerian or not). If $d_{+}(v)=0$ or $d_{-}(v)=0$, then the derivation yields 0. Moreover, we always have $v(\widetilde{G})-c(\widetilde{G})\leq v(G)-c(G)$. This can be compensated by a summation over $i$ if $d_{+}(v)=d_{-}(v)=1$ and both ends belong to the same edge.However, in this case the replica index $a$ is free and the summation over it yields a factor of $n$. Since we are only concerned with the order 0 in $n$, such a term does not contribute. Therefore, we conclude that the contribution of the first term of the LHS of \eqref{integral:eq} to $C^{'(k)}_{\widetilde{G}}(i_{1},..., i_{v(G)})$ and  $C^{''(k)}_{\widetilde{G}}(i_{1},..., i_{v(G)})$ also obey the bound.

In the case of the second operation, let us denote by $G_{1}$ and $G_{2}$ the two graphs on which the operation is performed and by $v_{1}$ the vertex of $G_{1}$ (resp. $v_{2}$ the vertex of $G_{2}$) where we derive with respect to $X_{i,a}$ (resp. $\overline{X_{i,a}}$) and by $\widetilde{G}$ the result of this operation. If $d_{-}(v_{1})=0$ or $d_{+}(v_{2})=0$. the derivative vanish. In the remaining cases, let us first observe that $v(\overline{G})-c(\overline{G})\leq (G_{1})-c(G_{1})+v(G_{2})-c(G_{2})$. Then, the following combinatorial lemma, whose proof is elementary, is helpful.
\begin{lemma} 
\label{lemmabound:lem}
If $G$ is a connected oriented graph, then $G$ is either Eulerian or has at least two vertices such that $d_{+}(v)\neq d_{-}(v)$.
\end{lemma}
As a consequence of this lemma, $\widetilde{G}$ is Eulerian if and only if $G_{1}$ and $G_{2}$ both are, since the operation only modifies the valence of the vertices $V_{1}$ and $V_{2}$. Moreover, in the particular case  $d_{-}(v_{1})=d_{+}(v_{2})=1$, the inequality is strict and  the extra power of $N$ cancels the power arising from the  summation over $i$.

However, it is not the graph $\widetilde{G}$ that contributes directly to the LHS of \eqref{integral:eq} to $C^{'(k)}_{\widetilde{G}}(i_{1},..., i_{v(G)})$ since some indices on the vertices of $G_{1}$ may be equal ton some indices on the vertices of $G_{2}$, leading to the identification of some vertices in $\widetilde{G}$. If we denote the resulting graph by $\overline{\widetilde{G}}$, then $v(\overline{\widetilde{G}})-c(\overline{\widetilde{G}})\leq v(\widetilde{G})-c(\widetilde{G})$. Moreover, an application of  lemma \ref{lemmabound:lem} shows that the inequality is strict if $\overline{\widetilde{G}}$ is Eulerian and $\widetilde{G}$ is not. Then, we may conclude that the contribution of the second operations to to $C^{'(k)}_{\overline{\widetilde{G}}}(i_{1},..., i_{v(G)})$ and  $C^{''(k)}_{\overline{\widetilde{G}}}(i_{1},..., i_{v(G)})$ also obey the bound.
\end{proof}

Finally, Theorem \ref{Wignerdependent:thm} follows from \eqref{Greenpotential:eq} since in the large $N$ limit the non Gaussian contribution to the Green function vanishes in the large $N$ limit, $N^{-1/2}\big[C^{''}_{\includegraphics[width=0.8cm]{cumulantii}}(1/z,i) \big]_{\text{order 0 in $n$}}\rightarrow 0$, as a consequence of proposition \ref{bound:pro}, order by order in $1/z$.

\subsection*{Acknowledgements}
The authors are partially supported by the grant ANR JCJC ``CombPhysMat2Tens".
AT is partially supported by the grant PN 
16 42 01 01/2016.

\noindent Thomas Krajewski, \href{mailto: Thomas.Krajewski@cpt.univ-mrs.fr}{\tt Thomas.Krajewski@cpt.univ-mrs.fr} \\
{\it\small  CPT, Aix-Marseille Universit\'e, Marseille, France, EU}



\noindent
Adrian Tanasa, \href{mailto:adrian.tanasa@labri.fr}{\tt ntanasa@u-bordeaux.fr}\\
{\it\small 
LaBRI, Universit\'e Bordeaux, Talence, France, EU}\\
{\it\small H. Hulubei Nat. Inst.  Phys.  Nucl. Engineering,
Magurele, Romania, EU}\\
{\it\small 
IUF Paris, France, EU }\\
\noindent Dinh Long Vu, \href{mailto:dinh-long.vu@polytechnique.edu}{\tt dinh-long.vu@polytechnique.edu}\\
{\it\small  \'Ecole Polytechnique, Palaiseau, France, EU}

\end{document}